%% file: SinghDistributedInference052019.tex
\documentclass[10pt, journal,twocolumn]{IEEEtran}
\pdfoutput=1

\usepackage[font=small]{caption}
\usepackage{amsmath}
\usepackage{mathtools}
\usepackage{amsmath}
\usepackage{amssymb}
\usepackage{cite}
\usepackage{array}
\usepackage{amsthm}
\usepackage{amsmath}
\usepackage{tabulary}
\usepackage{multirow}
\usepackage{hyperref}
\usepackage{subfig}
\IEEEoverridecommandlockouts
\usepackage{pifont}
\usepackage{epsfig}
\usepackage{graphicx}
\usepackage{url}
\usepackage{amssymb}
\usepackage{pgfplots}
\usepackage{bbm}
\usepackage{tabulary}
\usepackage{paralist} %
\usepackage{xcolor}
\graphicspath{ {./figs/} }

\include{notation}
\begin{document}
\title{On Provisioning Cellular Networks for Distributed Inference}

\author{\IEEEauthorblockN{Sarabjot Singh }

\thanks{ The author is with Uhana Inc. Palo Alto, CA.
}}
\maketitle

\begin{abstract}
Wireless traffic  attributable to  machine learning (ML) inference workloads is increasing with the proliferation of applications and smart wireless devices leveraging ML inference. Owing to  limited compute capabilities at these ``edge" devices, achieving  high inference accuracy often requires coordination   with a remote compute node or ``cloud" over the wireless cellular network. The accuracy of this distributed inference is, thus, impacted by the communication rate and reliability offered by the cellular network. In this paper, an analytical framework  is proposed to characterize  inference accuracy  as a function of   cellular network design. Using the developed framework, it is shown that cellular network should be provisioned with a minimum density of  access points (APs)  to guarantee a target inference accuracy, and the  inference accuracy achievable at asymptotically high AP density is limited by the air-interface   bandwidth.
Furthermore, the minimum accuracy required of  edge inference to deliver a target inference accuracy is   shown to be inversely proportional to the density of APs and the bandwidth.
\end{abstract}

\section{Introduction}
With the proliferation of  machine learning (ML) enabled applications  (e.g. Siri, GoogleHome, Alexa, etc.) and   wireless devices (e.g. smart speakers, phones,   wireless cameras), the traffic attributable  to such applications is bound to increase \cite{cisco}. These applications gather data, like voice commands or security videos, through  smart devices and run inference leveraging  ML model(s)\footnote{The term ``model"  refers to a parametric mapping function (e.g. neural network, decision tree, random forest, etc.) fitted using  a data driven training procedure.}, which are becoming increasingly computationally expensive. Owing to the limited compute, power, and  storage capabilities at these devices or ``edge", achieving high inference accuracy  with low delay   is challenging \cite{WuBCCCDHIJJLLLQ19}. As a result,   offloading of inference is attractive,   wherein a  partial (or full) inference may run on a  remote compute resource (with more computational power)  or ``cloud". Wireless cellular network, being the communication medium, plays a key role in enabling such a coordination between  edge and cloud.

Providing seamless end-user experience  for these applications creates unprecedented challenges for wireless network operators.   Cellular network's rate  and  reliability, being inherently bounded, has to be appropriately  provisioned to meet the delay and accuracy demands of this distributed inference paradigm. Moreover, the application developers also need to be cognizant of the uncertainty added by the communication delay to the inference accuracy and, in turn, application  performance, and design the edge model appropriately.

Most of the prior work in cellular network design for computation offload has focussed on the design of offloading strategies aiming to optimize for a myriad of objectives (see \cite{MaoYouZhaHuaLet17,MacBec17} for a survery).  The work in \cite{KoHanHua18} proposed a stochastic geometry based model for a wireless mobile edge computing network and characterized the average computation and communication latency as   function of network  parameters. The analysis was further extended by \cite{ParLee18},  which considered the impact of user association in a heterogeneous cellular network on the communication/compute delay distribution. The work in \cite{MukLee18} characterized the impact of these delays on energy consumption. None of the prior works, however,  characterized the impact of cellular network  design on  inference accuracy achievable in a distributed inference framework, and the consequent inter-play in  provisioning of cellular network and the edge  inference accuracy. This paper is aimed to bridge this gap.

This paper proposes a tractable approach to characterize the impact of cellular network design on application inference performance for distributed inference. In particular, application inference accuracy, as measured by mean squared error (MSE)\footnote{Inference MSE implies the inverse of accuracy throughout this paper.},  is derived as a function of the key system parameters: density of the wireless access points (APs), transmission bandwidth, edge and cloud model's accuracy. Using the developed framework,
\begin{itemize}
\item it is shown that average inference accuracy improves with increasing AP density, but saturates at a level inversely proportional to  the  bandwidth.
\item the minimum AP density  required to achieve a target inference accuracy for a given edge and cloud inference accuracy  is analytically characterized, and shown to   decrease with bandwidth.
\item the  minimum edge model accuracy required  to guarantee an overall target inference accuracy is derived and it is shown that higher AP density and/or bandwidth allows application developer to use a less accurate edge model, while guaranteeing the overall target inference accuracy.
\end{itemize}

\section{System Model}\label{sec:sysmodel}

The APs  are assumed to be distributed uniformly in $\real{2}$ as a homogeneous Poisson Point process (PPP)  $\PPP$ of density $\dnsty$. The devices in the network are assumed to be distributed according to an independent homogeneous PPP $\PPPu$ with density $\userdnsty$.    
The power received from an AP at $X \in \real{2}$ transmitting with power $\power{X}$ at a device at $Y \in \real{2}$  is $\power{X} H_{X,Y}\pl(X,Y)^{-1}$, where $H\in \real{+}$ is the fast fading power gain  and assumed to be Rayleigh distributed with unit average power, i.e.,  $H \sim \exp(1)$, and $\pl(X,Y) \triangleq \sha_{X,Y} \|X-Y\|^{4}$, where $\sha\in \real{+}$ denotes the large scale fading (or shadowing). Both small and large scale  fading are  assumed i.i.d across all device-AP pair.  The analysis in this paper is done for a \textit{typical} device located at the origin.   

\subsection{Uplink $\SINR$  and $\rate$}
Let  $\pl_{X}$  be the path loss  between the device at $X \in \real{2}$ and its serving AP. A full pathloss-inversion based power control  is assumed for   uplink transmission, where  a device at $X$ transmits with a power spectral density (dBm/Hz) $\power{X}=  \power{u} \pl_{X}$,  and $\power{u}$ is the open loop power spectral density.  Orthogonal access is assumed in the uplink  and    hence at any given resource block, there is at most one device transmitting in each cell.  Let $\PPPu^b$ be the point process  denoting the location of devices transmitting on the same resource as the typical device. The uplink $\SINR$ of the typical device (at origin) on a given resource block is
\begin{equation}\label{eq:sinr}
\SINR_u =\frac{ H_{0,\cc{0}}}{\SNR^{-1} + \sum_{X\in \PPPu^b}\pl_{X} H_{X,\cc{0}} \pl(X,\cc{0})^{-1}},
\end{equation}
where $\cc{0}$ denotes the  AP serving the typical device, $\SNR \triangleq \frac{\power{u}  \mathrm{L_0}}{{\mathrm{N}_0}}$ with $\mathrm{N}_0$ being the thermal noise spectral density, and $\mathrm{L_0}$ is the free space path loss at a reference distance. Every device is assumed to be using minimum path loss for association and is assumed that each AP has at least one device   with data to transmit in uplink.  Assuming an equal partitioning  of the total  uplink resources at an AP among the associated uplink  users, the  uplink rate of the typical device is
\begin{equation}\label{eq:ratemodel}
\rate_u = \frac{\res}{\load{}}\log\left(1+\SINR_u\right),
\end{equation}
where $\res$ is the uplink bandwidth,   $\load{}$ denotes the total number of devices served by the AP. Along similar lines (as in \cite{SinZhaAnd14}), downlink rate $\rate_d$ can be defined.

\subsection{Inference Framework}
An inference framework  is assumed wherein, for each inference input (e.g. a chunk of speech or an image), denoted by $x$, the device transmits the inference input  to the cloud while,  concurrently,  computing a local inference output (say $y_{d}$) using the edge model. If the device receives the inference result  from the cloud (say  $y_{c}$)    within target delay budget (denoted by $\del{t}$), it is used as the final output  $y_{o}$;   otherwise the device uses the edge model's output $y_{d}$. Therefore, 
\begin{equation}\label{eq:inferenceoutput}
y_o  = \begin{cases}
	y_c, & \text{ if }  \delrnd \leq \del{t},\\
	y_{d}, & \text{ otherwise, }
	\end{cases}
\end{equation}
where $\delrnd$ denotes the cumulative delay incurred in receiving $y_{c}$ at the device. 
Assuming inference input and output has a fixed (over the air) payload size, i.e.   $|x| = |y| = q/2$, the cloud inference delay is
\begin{equation}\label{eq:delay}
\delrnd  =  \frac{q}{2\rate_u} + \frac{q}{2\rate_d}  + \del{c},
\end{equation}
where the first two terms correspond to communication delays (in uplink and downlink respectively) and $\del{c}$ is   the compute delay (assumed to be fixed\footnote{the communication delay associated with the cloud transport network is assumed to be incorporated in $\del{c}$.}) incurred by the cloud inference model.  

Denoting the actual inference output by $y$, device and cloud model inference accuracy are defined by their mean square errors (MSE's), 
\[\mse{d} = \expect{(y-y_d)^2} \text{  and } \mse{c} = \expect{(y-y_c)^2},\] respectively, where the expectation is over the data distribution, and cloud model's accuracy is assumed to be more accurate than that of edge model, i.e.,  $\mse{c}\leq \mse{d} $.
 
As a result of the inference mechanism in (\ref{eq:inferenceoutput}), the average  MSE (denoted by $\avmse{}$) for a typical device  is
\begin{equation}\label{eq:mse}
\avmse{} \triangleq \expect{ \indic( \delrnd \leq \del{t})\mse{c}} + \expect{\indic(\delrnd > \del{t})\mse{d}},
\end{equation}
where $\indic(A)$ denotes the indicator of the event $A$.

 The notation used in this paper  is summarized in Table \ref{tbl:param}.
\begin{table}
	\centering
	\caption{Notation and simulation parameters}
	\label{tbl:param}
	\begin{tabulary}{\columnwidth}{ |C | L | L|}
		\hline
		\textbf{ Notation} & \textbf{Parameter} & \textbf{Value (if not specified }\\\hline
		$\SINR_x$, $\rate_x$ & Uplink (x:u), downlink (x:d) $\SINR$ and $\rate$   &  \\\hline
		$\mse{d}$, $\mse{c}$, $\mse{t}$, $\mse{}$ & Edge, cloud, target, and average inference MSE  &  $\mse{d} = 1.5\mse{c}$\\\hline
		$\dnsty$, $\userdnsty$ & density of APs and devices  & \\\hline
		$\res$ &  uplink transmission bandwidth &  \\  \hline
		$\load{}$ & number of uplink devices in AP serving the typical device  & \\\hline
		$\del{t}$, $\del{c}$, $\delrnd$ & target delay budget, cloud compute delay, and cloud inference delay &  \\ \hline
		$\querysize$ & cumulative (uplink and downlink) size  of cloud inference input and output &   \\\hline
		$\infrate$ & inference rate & $\frac{q}{\res(\del{t}-\del{c})}$\\\hline
	\end{tabulary}
\end{table}

\section{Inference Accuracy}\label{sec:analysis}
The accuracy of the distributed inference model is characterized as a function of the network parameters in this section. The following two assumptions are taken to simplify the analysis. 
\begin{asmptn}\label{asmptn:rate}
The downlink rate of a typical device is assumed to be equal to that of the uplink. 
\end{asmptn}
In \cite{SinZhaAnd14}, it was shown that downlink rate stochastically dominates uplink rate, hence the above assumption leads to over-estimation of the cloud inference delay in (\ref{eq:delay}), which simplifies to
\begin{equation}\label{eq:delayrnd}
\delrnd  =  \frac{q}{\rate_u}   + \del{c}.
\end{equation}
Using (\ref{eq:inferenceoutput}) and (\ref{eq:delayrnd}), the minimum uplink rate required for a device to use cloud inference, i.e. $y_o = y_c$, is  $\frac{\querysize}{\del{t}-\del{c}}$. Henceforth, this minimum rate normalized by the  transmission bandwidth, $\infrate \triangleq \frac{\querysize}{\res(\del{t}-\del{c})}$,  is referred to as the \textit{inference rate}.
\begin{asmptn}
The load on the AP serving the typical device is assumed to be constant and equal to its average value (denoted by $\avload{}$), i.e.,  
\begin{equation}\label{eq:avload}
\load{} \approx \avload{} = 1+\frac{1.28}{\ndnsty},
 \end{equation}
where $\ndnsty \triangleq \dnsty/\userdnsty$. 
\end{asmptn}
This assumption was taken in past works (see \cite{SinZhaAnd14} and references therein)  without loss of generality of design insights.

\begin{lem}\label{lem:delaydist}
\textbf{Delay distribution.}
The cloud inference delay distribution experienced by a typical device is  
\begin{multline}
\dcov(d) \triangleq   \pr(\delrnd \leq \del{}) = \exp\left(-\hgc\circ{\srtransform\left(\avload{}\infrate\frac{\del{t}-\del{c}}{\del{}-\del{c}}\right)} \right)\\ \, \, \forall \del{} > \del{c},
\end{multline}
where $\hgc(x) = \sqrt{x}\arctan(\sqrt{x})$ and  $\srtransform(x) = 2^{x}-1$. 
\end{lem}
\begin{proof}
See Appendix \ref{sec:proofdelaydist}.
\end{proof}
 Since both $\hgc$ and $\srtransform$ are monotonically increasing function, cloud inference delay is   proportional to the   inference rate and average load.  
 
\begin{lem}\label{lem:avmse}
\textbf{Average MSE.}
The average output inference MSE of  a typical device is 
\begin{equation}
\avmse{} = \mse{d} - (\mse{d}-\mse{c})\exp\left(-\hgc\circ{\srtransform\left(\avload{}\infrate \right)}\right).
\end{equation}
\end{lem}
\begin{IEEEproof}
Follows by using  Lemma  \ref{lem:delaydist} with (\ref{eq:mse}).
\end{IEEEproof}
As can be observed from Lemma \ref{lem:avmse}, $\mse{c} \leq \avmse{} \leq \mse{d}$.
The term on the right captures the average MSE improvement provided by the cloud inference. This improvement diminishes with increasing  inference rate ($\infrate$) and higher device load ($\avload{}$).  The following section further formalizes these insights. 

\section{Performance analysis and insights}
 Figure \ref{fig:msevsdensity} shows the variation of  average inference MSE normalized by the cloud MSE (i.e. $ \avmse{}/\mse{c}$) with the normalized  AP density ($\ndnsty$) and inference rate\footnote{the values of system parameters not specified explicitly are as per Table \ref{tbl:param}.}. As observed, for a given AP density average inference MSE increases with $\infrate$ (or decreases with bandwidth) and approaches that of the edge, as with the increase in communication delay device has to rely on the edge model output. Moreover, average MSE decreases with increasing AP density, but saturates eventually. This is formalized with the following.
 \begin{figure}
	\centering
\includegraphics[scale=1, width=0.95\columnwidth]{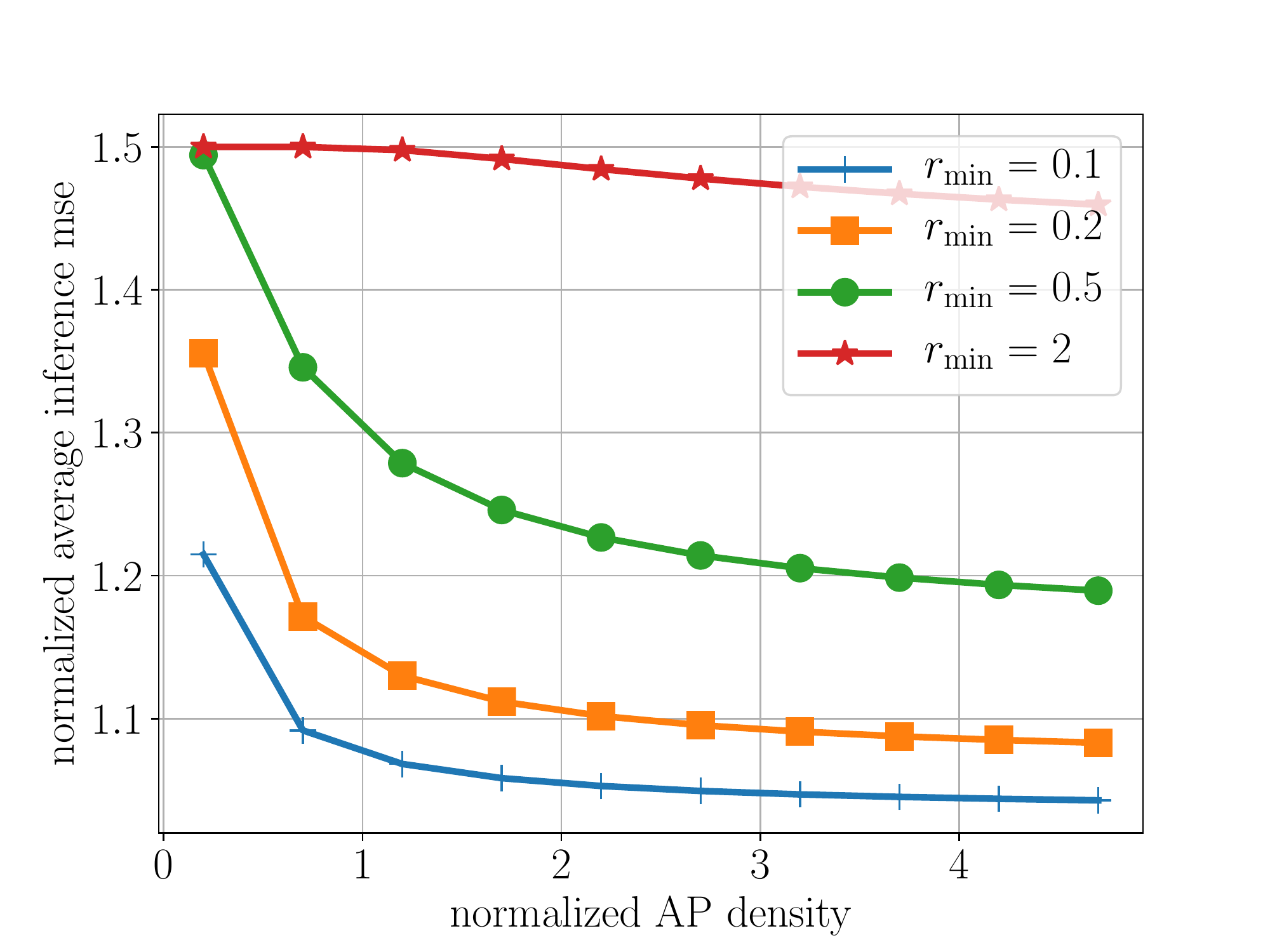} 
		\caption{Variation of average inference MSE with AP density at different inference rates.}
	\label{fig:msevsdensity}
\end{figure}
\begin{figure}
	\centering
\includegraphics[scale=1, width=0.95\columnwidth]{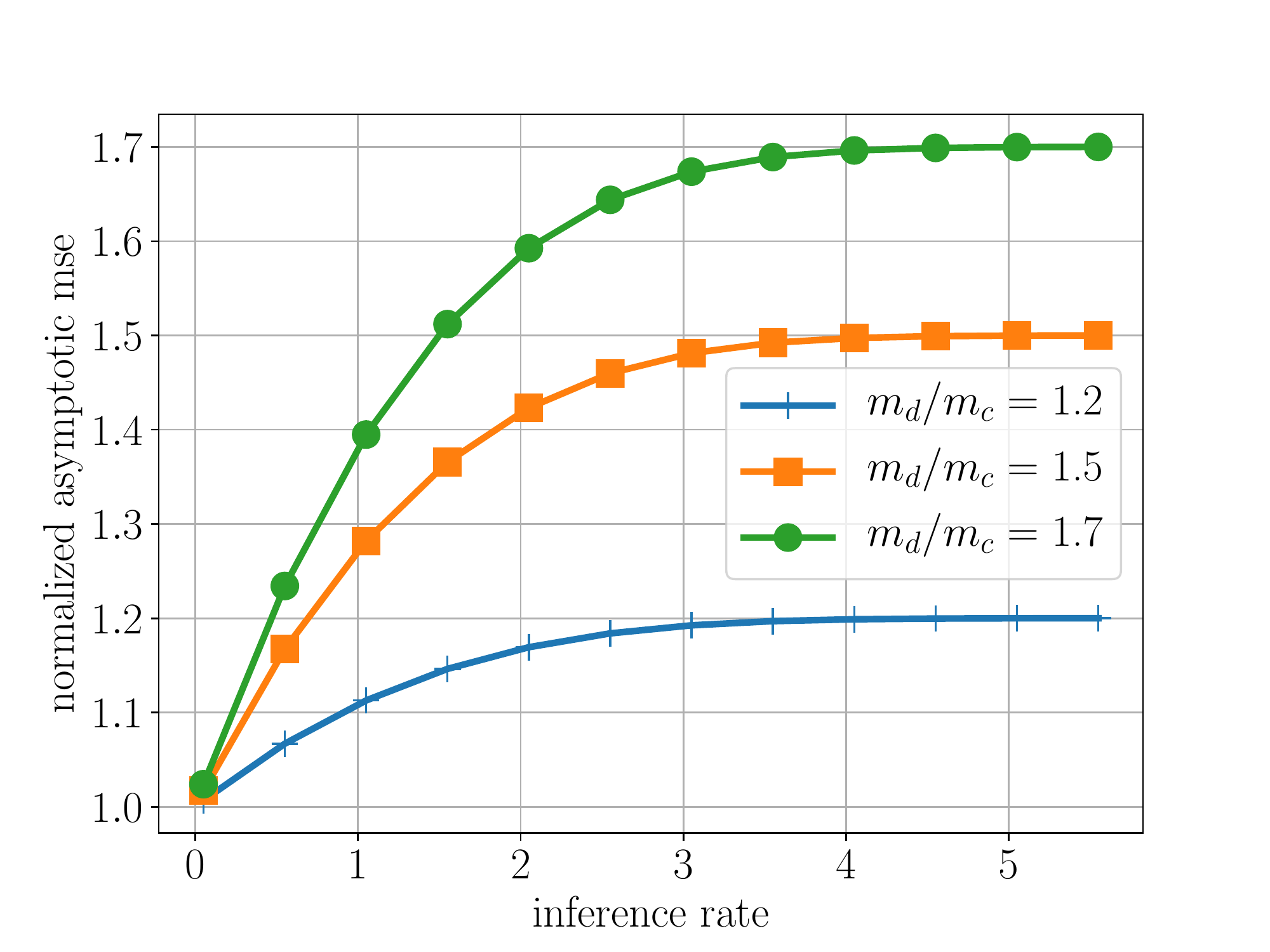} 
		\caption{Variation of asymptotic MSE with inference rate for various edge MSEs.}
	\label{fig:mtvsrate}
\end{figure}
\begin{figure} 
	\centering
\includegraphics[scale=1, width=0.95\columnwidth]{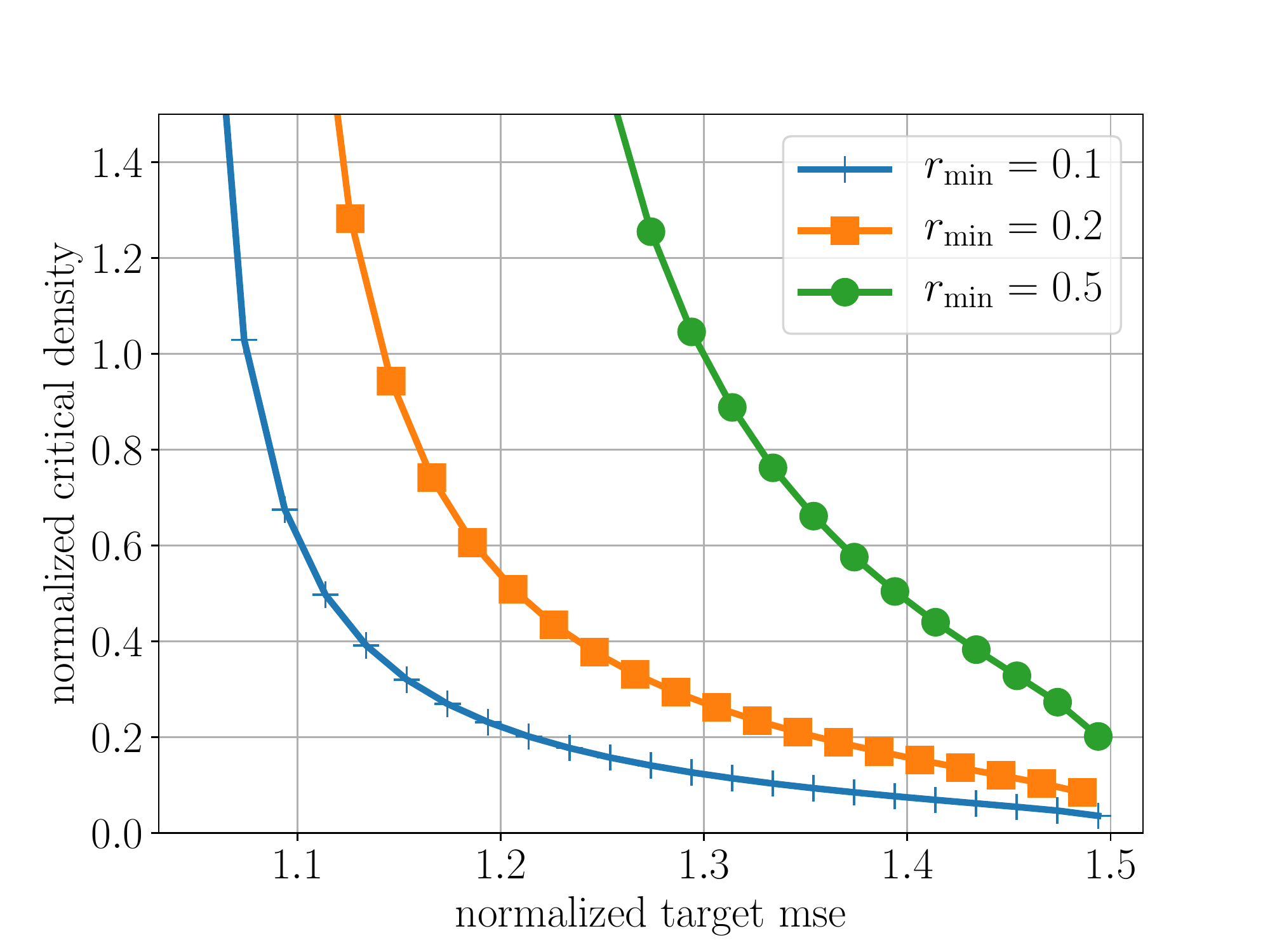} 
		\caption{Variation in critical density with the target MSE for different inference rates.}
	\label{fig:dnstyvsmt}
\end{figure}
\begin{figure} 
	\centering
	\includegraphics[scale=1, width=0.95\columnwidth]{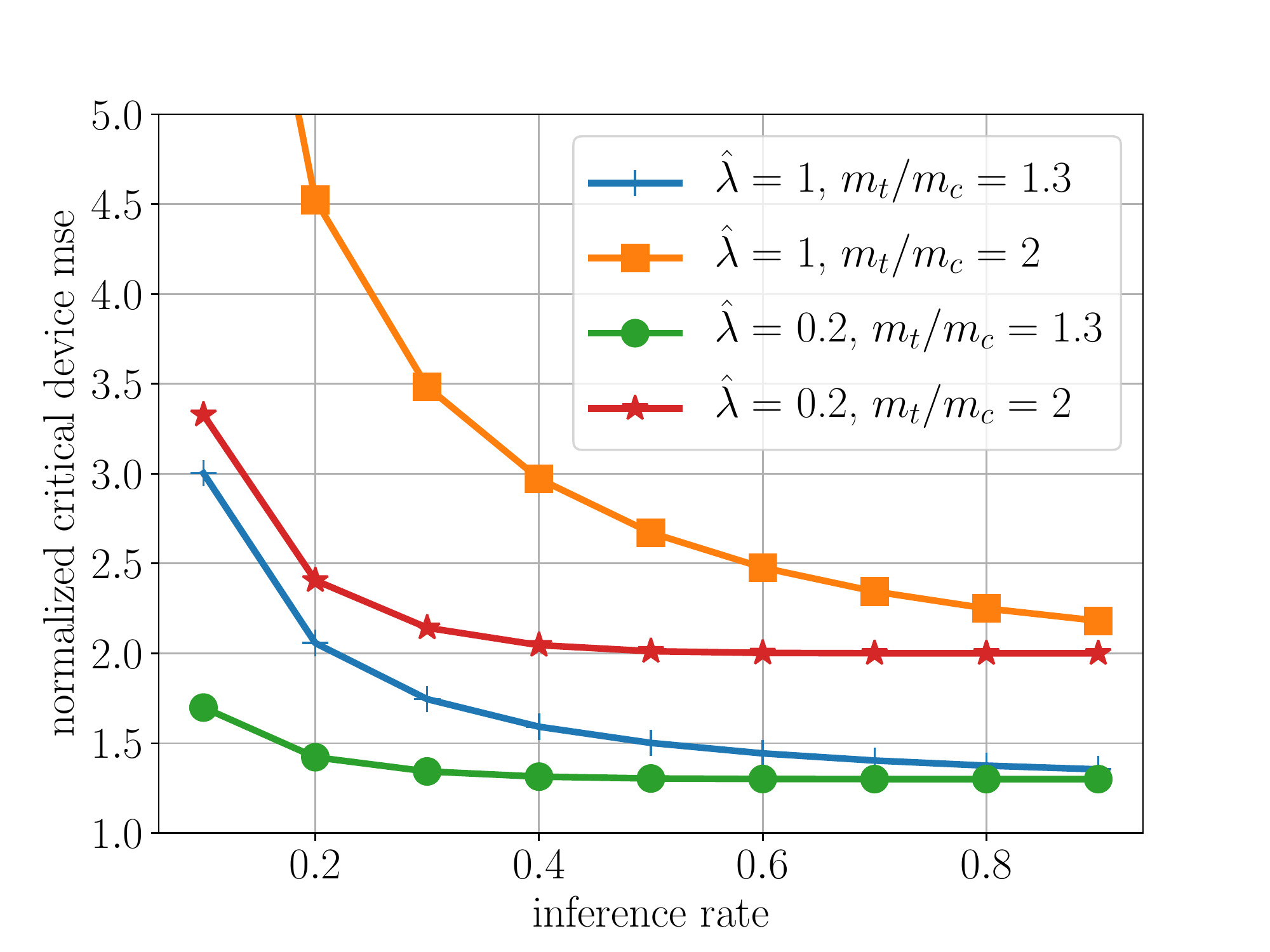} 
	\caption{Variation of critical edge MSE with inference rate for different AP densities and target MSEs.}
	\label{fig:me_nvsr_t}
\end{figure}

\begin{cor}\label{cor:minmse}\textbf{Asymptotic MSE.}  The average output MSE  at asymptotically high AP density is 
\begin{align*}
\mse{\mathrm{asy}} & \triangleq \lim_{\dnsty \to \infty} \avmse{} = 
 \mse{d} - (\mse{d}-\mse{c})\exp\left(-\hgc\circ{\srtransform\left(\infrate\right)} \right).
\end{align*}
\end{cor}
\begin{proof}
Follows by replacing $\lim_{\dnsty\to\infty}\avload{} = 1$ in Lemma \ref{lem:avmse}.
\end{proof}
This shows that   at high AP density, the gain in MSE from cloud inference is limited by the inference rate, as even if the entire bandwidth is allocated  to a single device, the cloud inference delay may not be lower than the delay budget.  However, increasing bandwidth or decreasing inference rate reduces this MSE.
\begin{align*}
\text{Also }
\lim_{\infrate \to \infty}\mse{\mathrm{asy}}  = \mse{d}   \text{ and }
\lim_{\infrate \to 0}\mse{\mathrm{asy}}   = \mse{c}.
\end{align*} At low inference rate, the asymptotic  MSE  approaches that of the cloud model; and that of the edge at  high inference rate.  
 Figure \ref{fig:mtvsrate} shows the variation of normalized asymptotic MSE ($\mse{\mathrm{asy}}/\mse{c}$) between these extremes.

\begin{cor}\label{cor:cricdnsty}\textbf{Critical density.} The minimum AP density 
to guarantee   an inference accuracy ($\mse{t}$) is   
\small
	\begin{multline}	\dnsty_c  \triangleq 1.28 \userdnsty \left\{ \frac{1}{\infrate}\log_2\left(1 + \hgc^{-1}\left[\log\left( \frac{\mse{d} -\mse{c}}{\mse{d}- \mse{t}}\right)\right]\right)   -1 \right\}^{-1},\\
\forall\,\, \mse{t} > \mse{\mathrm{asy}}. 
\end{multline}
\normalsize
\end{cor}
\begin{proof}
See Appendix \ref{sec:proofcricdnsty}.
\end{proof}
 Figure \ref{fig:dnstyvsmt} shows the variation of normalized critical density ($\dnsty_c/\userdnsty$) with the  normalized target MSE $\mse{t}/\mse{c} \,\forall\,\, \mse{t}   \text{ s.t. } \mse{\mathrm{asy}} < \mse{t} \leq \mse{d}$. As can be seen, the network needs to be provisioned with higher AP density with increasing accuracy demands. Moreover,  critical density  also increases with increase in inference rate (or decrease in bandwidth) for a target MSE.

As highlighted by Corollary \ref{cor:minmse},  even at high infrastructure density  the gain from cloud inference may be limited -- guaranteeing  an overall accuracy, thus, requires a minimum accuracy from the edge inference model.
\begin{cor}\label{cor:cricMSE}\textbf{Critical edge MSE.} The maximum allowed device MSE required to guarantee an overall MSE lower than $m_t$   is 
	\begin{align}
	\mse{d, \mathrm{max}} \triangleq \mse{c}  \frac{\mse{t}/\mse{c} - \exp\left(-\hgc\circ{\srtransform\left(\avload{}\infrate \right)}\right) }{1 - \exp\left(-\hgc\circ{\srtransform\left(\avload{}\infrate \right)}\right)}
	\end{align}
\end{cor}
\begin{proof}
See Appendix \ref{sec:proofcricdnsty}.
\end{proof}

From above, it is evident that as $
\avload{}\infrate\to \infty$,   $\mse{d,\mathrm{max}} \to m_t $, 
or as the cloud inference delay increases beyond the target delay, the edge  inference accuracy needs to be at least at par  with the target accuracy. Figure \ref{fig:me_nvsr_t} shows the variation of normalized critical edge MSE ($ \mse{d,\mathrm{max}}/\mse{c}$) for different system parameters. It can be seen that, for any given inference rate, higher AP density allows edge inference model to be less accurate, while meeting the overall target inference accuracy. Corollary \ref{cor:cricMSE} and the aforementioned  insights emphasize the constraints imposed by the network parameters on the design of edge inference model.


\section{Conclusion}
This paper proposes an analytical framework for characterizing the performance of   distributed inference in wireless cellular networks. To the author's best knowledge, this is the first work to present   the trade-offs involved in the co-design of cellular networks and distributed inference. This work can be extended to analyze and compare different policies for offloading inference and their impact on network design. The network model can be extended to analyze the impact of congestion in the cloud transport network (as in {\cite{KoHanHua18}) and user association in heterogeneous cellular networks (as in \cite{ParLee18}) on  distributed inference.

\appendices
\section{}\label{sec:proofdelaydist}
\begin{IEEEproof}[Derivation of cloud inference delay distribution]
\begin{align*}
&\pr(\delrnd \leq d)  = \pr\left(\rate_u \geq \frac{q}{\del{} - \del{c}}\right)\\
& \overset{(a)}= \pr\left(\SINR_u \geq \srtransform\left(\frac{\avload{}q}{\res(\del{} - \del{c})}\right) \right) \\
 & \overset{(b)}=  \exp\left(-\hgc\circ{\srtransform\left(\frac{\avload{}q}{\res(\del{}-\del{c})}\right)}\right), 
 \end{align*}
where (a) follows using (\ref{eq:ratemodel}) and (\ref{eq:avload})  and  $\srtransform(x) = 2^{x}-1$; and (b) follows by using uplink $\SINR$ distribution from \cite{SinZhaAnd14}.
\end{IEEEproof}

\section{}\label{sec:proofcricdnsty}
\begin{IEEEproof}[Derivation of critical density] For the average inference MSE to be less than a target, i.e., 
\begin{multline}
\avmse{} \leq \mse{t} \\\overset{(a)}{\implies} \mse{d} - (\mse{d}-\mse{c})\exp\left(-\hgc\circ{\srtransform\left(\avload{}\infrate \right)}\right) \leq \mse{t}, \\\text{ or } 
\srtransform\left(\avload{}\infrate \right) \leq \hgc^{-1}\left(\log\left( \frac{\mse{d} -\mse{c}}{\mse{d}- \mse{t}}\right)\right)  ,
\end{multline} where (a) follows using Lemma \ref{lem:avmse}. The critical density is arrived at by replacing (\ref{eq:avload}) in above.
\end{IEEEproof}
\begin{IEEEproof}[Derivation of critical edge MSE] Follows by algebraic manipulation on (a) above.\end{IEEEproof}

\bibliographystyle{ieeetr}
\bibliography{IEEEabrv,/Users/sarabjotsingh/Dropbox/research/refoffload}

\end{document}

%% file: notation.tex
\newcommand{\expect}[1]{{\mathbb{E}\left[{#1}\right]}}

\newcommand{\pr}{{\mathbb{P}}}

\newcommand{\infrate}{r_\mathrm{min}}
\newcommand{\hgc}{\mathrm{C}}
\newcommand{\srtransform}{\mathrm{T}}
\newcommand{\querysize}{q}

\newcommand{\power}[1]{\mathrm{P}_{#1}}
\newcommand{\res}{b}

\newcommand{\dnsty}{{\lambda}}
\newcommand{\ndnsty}{{\hat{\lambda}}}

\newcommand{\mse}[1]{m_{#1}}
\newcommand{\avmse}[1]{\bar{m}_{#1}}

\newcommand{\userdnsty}{\lambda_u}

\newcommand{\real}[1]{\mathbb{R}^{#1}}

\newcommand{\cc}[1]{\mathcal{B}_{#1}}

\newcommand{\SINR}{\mathtt{SINR}}
\newcommand{\pl}{L}

\newcommand{\SNR}{\mathtt{SNR}}

\newcommand{\delrnd}{D}

\newcommand{\del}[1]{d_{#1}}

\newcommand{\rate}{\mathtt{Rate}}

\newcommand{\PPP}{\Phi}

\newcommand{\PPPu}{\Phi_{u}}

\newcommand{\load}[1]{{N}_{#1}}

\newcommand{\avload}[1]{\bar{n}_{#1}}

\newcommand{\sha}{S}

\newcommand{\indic}{\mathbbm{1}}

\newcommand{\dcov}{\mathcal{D}}

\newtheorem{cor}{Corollary}

\newtheorem{lem}{Lemma}

\theoremstyle{definition}

\theoremstyle{thm}
\newtheorem{asmptn}{Assumption}